\documentclass{fundam}
\usepackage{url} 
\usepackage[ruled,lined,noend,noline,linesnumbered,nofillcomment]{algorithm2e}

\DontPrintSemicolon
\usepackage{graphicx}

\usepackage{amsfonts}
\usepackage{float}
\usepackage{framed}
\usepackage{caption}
\usepackage{comment}
\usepackage{booktabs} 
\usepackage{siunitx} 
\usepackage{multirow}
\usepackage{diagbox}

\usepackage{hyperref}

\begin{document}

\setcounter{page}{257}
\publyear{2021}
\papernumber{2073}
\volume{182}
\issue{3}

  \finalVersionForARXIV

\title{Efficient Algorithms for Maximum Induced Matching Problem in Permutation and Trapezoid Graphs}


\author{Viet Dung Nguyen, Ba Thai Pham, Phan Thuan Do\thanks{Address  for correspondence:
       Hanoi University of Science and Technology, 1 Dai Co Viet, Hai Ba Trung, Ha Noi, Viet Nam.\newline \newline
          \vspace*{-6mm}{\scriptsize{Received  July 2021; \ revised July 2021.}}}
       \\
Hanoi University of Science and Technology \\
1 Dai Co Viet, Hai Ba Trung, Ha Noi, Viet Nam\\
 dungnv@soict.hust.edu.vn,$\,$ thai.pb144038@sis.hust.edu.vn,$\,$
 thuandp@soict.hust.edu.vn
}

 \maketitle

\runninghead{V.D. Nguyen et al.}{Maximum Induced Matching in Permutation and Trapezoid Graphs}

\begin{abstract}
  We first design an $\mathcal{O}(n^2)$ solution for finding a maximum induced matching in permutation graphs given their permutation models, based on a dynamic programming algorithm with the aid of the sweep line technique. With the support of the disjoint-set data structure, we improve the complexity to $\mathcal{O}(m + n)$. Consequently, we extend this result to give an $\mathcal{O}(m + n)$ algorithm for the same problem in trapezoid graphs. By combining our algorithms with the current best graph identification algorithms, we can solve the MIM problem in permutation and trapezoid graphs in linear and $\mathcal{O}(n^2)$ time, respectively. Our results are far better than the best known $\mathcal{O}(mn)$ algorithm for the maximum induced matching problem in both graph classes, which was proposed by Habib et al.
\end{abstract}


\section{Introduction} \label{S:introduction}

	The maximum matching problem is one of the most fundamental and applicable problems in graph theory. Given a graph $G = (V, E)$, a maximum matching is a subset $M \subseteq E$ of maximum size so that every two distinct edges in $M$ do not share a common vertex. Its applications can be found everywhere, from VLSI circuit design \cite{dagan1988trapezoid} to archaeology and chemistry \cite{golumbic2004algorithmic}. The best-known algorithm for the maximum matching problem in general graphs is $\mathcal{O}(\sqrt{V}E)$ \cite{micali1980v}. However, we can achieve better running time on many special graph classes thanks to their particular properties.
	
	If every two distinct edges in a matching $M$ are not connected by an edge in $G$, then $M$ is called an \textit{induced matching}. Recently, the maximum induced matching (MIM) problem has drawn enormous attention among researchers because of its importance in many fields such as artificial intelligence (cooperative path-finding problem \cite{surynek2014compact}, neural information processing \cite{kocaoglu2017experimental}), VLSI design \cite{golumbic2000new}, and marriage problems \cite{stockmeyer1982np}. The problem is proved to be NP-hard in general graphs \cite{stockmeyer1982np}. Some exponential time algorithms for the MIM problem in general graphs are proposed recently by Chang et al. \cite{chang2015moderately} and Xiao et al. \cite{XT17}. Besides, it is known for polynomial-time maximum induced matching on special graph classes such as co-comparability graphs (including circular-arc graphs \cite{golumbic1993irredundancy}, interval graphs \cite{cameron1989induced}, etc.), circular-convex bipartite and triad-convex bipartite graphs \cite{pandey2017induced}, AT-free graphs \cite{kohler1999graphs} and hexagonal graphs \cite{ervevs2016maximum}. In chordal graphs, finding a MIM can be done in linear time \cite{brandstadt2008maximum}.
	
	Dagan et al. \cite{dagan1988trapezoid} introduced trapezoid graphs in 1988. Given two parallel horizontal lines, a trapezoid is formed by two points on the upper line and two points on the lower line. A trapezoid graph is an intersection graph, \textit{i.e.}, a graph representing the pattern of intersections of a family of sets, built from such a set of trapezoids. A trapezoid representation (or trapezoid model) of a trapezoid graph includes two such horizontal lines and the set of trapezoids, which is used to forms that trapezoid graph. A permutation graph is a particular case of trapezoid graphs, in which the two intervals that define each trapezoid in the trapezoid representation shrink into two points only. Therefore, a permutation representation can be represented as a permutation of the first $n$ positive integers.
	
	Permutation graphs, as well as trapezoid graphs, are weakly chordal graphs \cite{cameron2003finding}. They are also Asteroidal-Triple-free (AT-free) graphs. In \cite{kohler1999graphs}, by applying the result from \cite{broersma1999independent}, MIM was solved in polynomial time for AT-free graphs. In \cite{chang2003induced}, the authors proposed a linear time algorithm for MIM on bipartite permutation graphs (which are bipartite AT-free graphs).
	
	Denote $n$ and $m$ as the number of vertices and edges in a graph, respectively. Do et al. \cite{do2017efficient} introduced an efficient algorithm to find a maximum matching in trapezoid graphs in $\mathcal{O}(n(\log n)^2)$. On the other hand, Rhee et al. \cite{rhee1995finding} proposed an $\mathcal{O}(n\log\log n)$ algorithm for such problems in permutation graphs. However, to the best of our knowledge, the MIM problem has not been mentioned with specific algorithms on permutation graphs. The best-known algorithm for finding a weighted induced matching for all co-comparability graphs, which are a superclass of permutation graphs and trapezoid graphs, has time complexity of $\mathcal{O}(mn)$ \cite{habib2020maximum}. So far, there is no other superclass of permutation graphs and trapezoid graphs on which the MIM problem is proved to be solvable in faster time than $\mathcal{O}(mn)$. Consequently, we can consider $\mathcal{O}(mn)$ the fastest time complexity to find a MIM on a permutation graph or a trapezoid graph.
	
	In this paper, we introduce more efficient algorithms for the MIM problem in both permutation graphs and trapezoid graphs. We first design in Section 2 an $\mathcal{O}(n^2)$ algorithm and then an $\mathcal{O}(m\log\log n + n)$ algorithm for finding a MIM in permutation graphs given their permutation models. These algorithms are based on a dynamic programming method with the aid of the sweep line technique on a geometry representation of permutation graphs. Our approach is to construct the longest chain of ordered edges, which form an induced matching. A sweep line moving from right to left correctly determines the order of dynamic processes. The edge set of the given permutation graph can be built from the vertex set of that graph in $\mathcal{O}(m + n)$ time by employing the benefits of the linked list data structure. Especially, with the support of the disjoint-set data structure, we improve the overall running time for finding a MIM to $\mathcal{O}(m + n)$ time in permutation graphs. Furthermore, in Section 3, we generalize this algorithm with the same running time $\mathcal{O}(m + n)$ to trapezoid graphs given trapezoid models. With the combination of our algorithms and the current best graph identification algorithms that generate presentation models for permutation graphs and trapezoid graphs, we can solve the MIM problem in permutation and trapezoid graphs in linear and $\mathcal{O}(n^2)$ time, respectively. Our results are far better than the best known $\mathcal{O}(mn)$ algorithm \cite{habib2020maximum} in both graph classes. This paper is the complete version including preliminary results of our conference papers \cite{nguyen2018dynamic} and \cite{nguyen2019quadratic}.

\section{Fundamental definitions} \label{S:Definitions}

\subsection{Maximum induced matching}

    A subset $M \subseteq E$ is an induced matching of graph $G = (V, E)$ if for every two distinct edges $e_1 = u_1v_1$ and $e_2 = u_2v_2$ in $M$ we have $u_1u_2 \notin E$, $u_1v_2 \notin E$, $v_1u_2 \notin E$ and $v_1v_2 \notin E$. We denote $L(G)$ as the line graph of $G$, \textit{i.e.}, each vertex of $L(G)$ represents an edge of $G$. Two vertices in $L(G)$ are adjacent if and only if their corresponding edges share a common endpoint in $G$. We also denote $G^2$ as the graph having the same vertex set as $G$ and two vertices are adjacent if their distance, \textit{i.e.}, the number of edges in a shortest path connecting them, is at most 2 in $G$. Then, the problem to find an induced matching of maximum cardinality is exactly the maximum independent set problem on $L(G)^2$.

\subsection{Permutation graph}

	Let $\pi = (\pi(1), \pi(2), ..., \pi(n))$ be a permutation of the first $n$ positive integers. We build an undirected graph $G(\pi) = (V, E)$ in which the vertex set $V = \{1, 2, ..., n\}$ and an edge $uv \in E$ if and only if $(u - v)(\pi^{-1}(u) - \pi^{-1}(v)) < 0$, where $\pi^{-1}(i)$ is the position of $i$ in $\pi$. An undirected graph $G$ is called a \textit{permutation graph} if there is a permutation $\pi$ such that $G(\pi)$ is isomorphic to $G$ (see Figure \ref{F:permutation graph example} for an example), and $\pi$ is called a \textit{permutation representation} (or \textit{permutation model}) of $G$.

\begin{figure}[!ht]
\centering
\includegraphics[width=13cm]{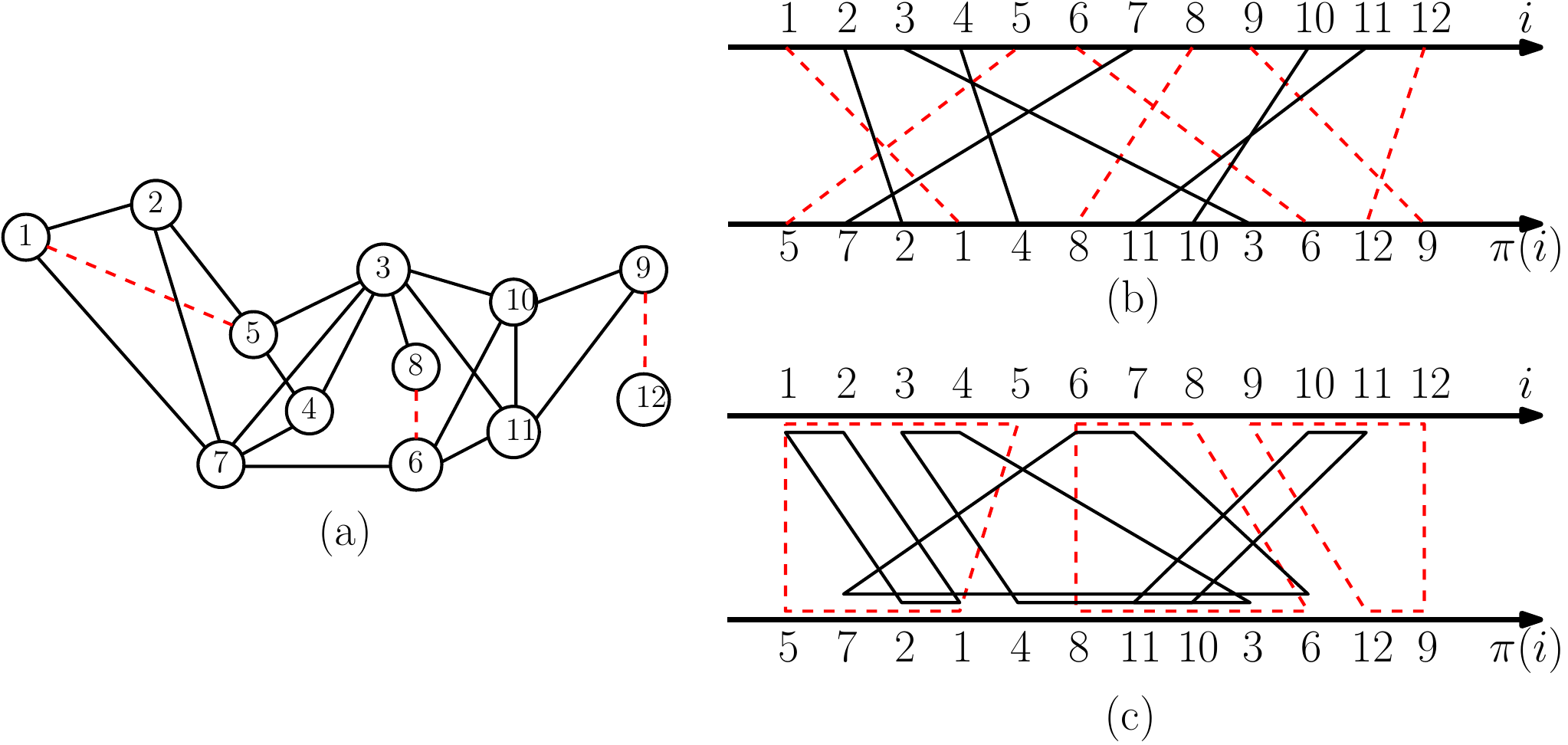}\vspace*{-1mm}
\caption{\rm A permutation graph $G$ (a) and one of its corresponding permutation representation (b), which has $\pi = (5, 7, 2, 1, 4, 8, 11, 10, 3, 6, 12, 9)$. A trapezoid model of the trapezoid graph $L(G)^2$, partially shown in (c), can be constructed from the permutation representation of $G$. A maximum induced matching for $G$ is $M = \{(1, 5), (6, 8), (9, 12)\}$, also seen as a maximum independent set of $L(G)^2$.}
\label{F:permutation graph example}
\end{figure}

\begin{figure}[!h]
\vspace*{-2mm}
\centering
\includegraphics[width=12.9cm]{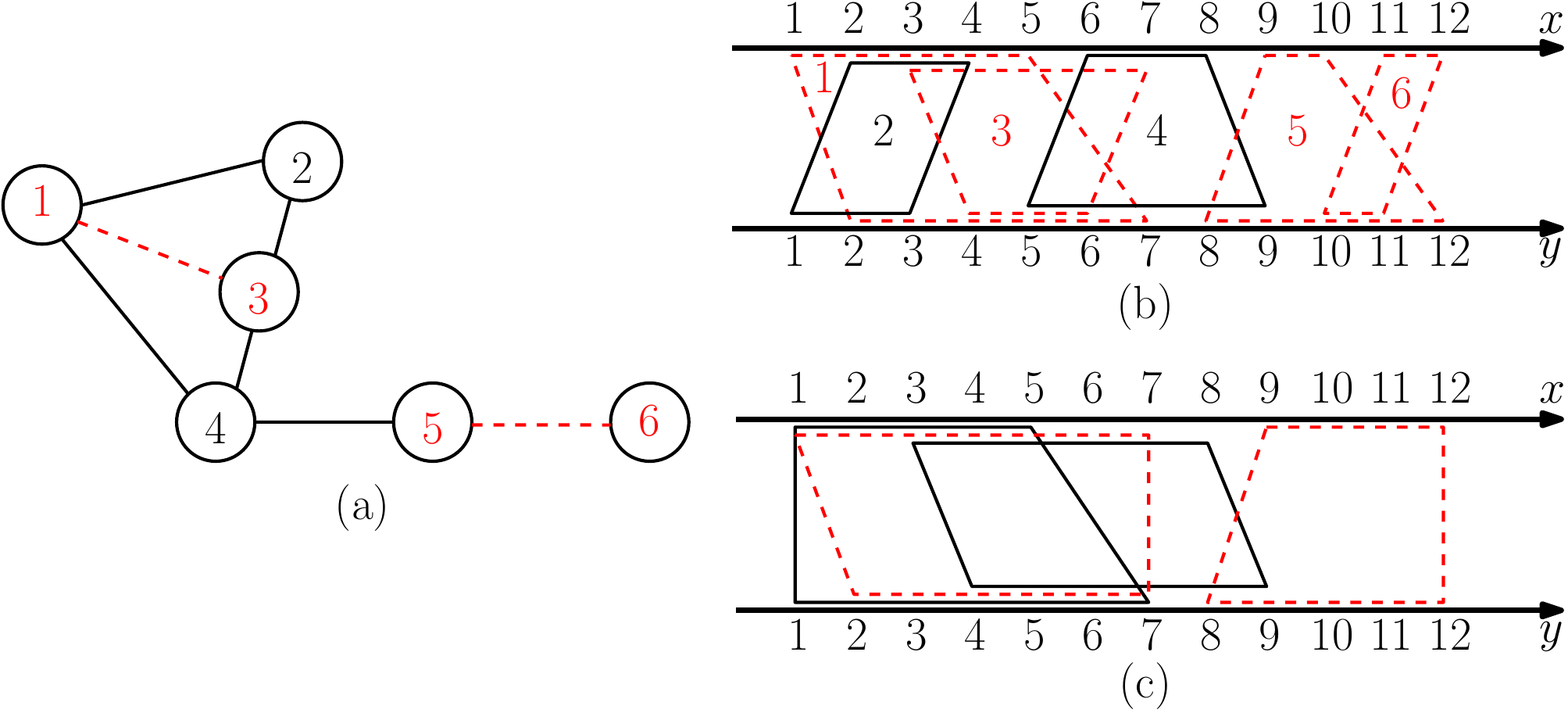}\vspace*{-1mm}
\caption{\rm A trapezoid graph $G$ (a) and one of its corresponding trapezoid representations (b). A trapezoid model of $L(G)^2$, partially shown in (c), can be constructed from the trapezoid representation of $G$. A maximum induced matching for $G$ is $M = \{(1, 3), (5, 6)\}$, also seen as a maximum independent set of $L(G)^2$.}
\label{F:trapezoid graph example}\vspace*{-2mm}
\end{figure}

\subsection{Trapezoid graph}

	We assume that a trapezoid model is given by a set of trapezoids $\tau$. Given two parallel horizontal axes: $x$-axis and $y$-axis, each trapezoid $A \in \tau$ is given by four values $x_1, x_2, y_1, y_2$ ($x_1 \le x_2$, $y_1 \le y_2$) such that [$x_1, x_2$] and [$y_1, y_2$] are the two intervals on these two axes that form the trapezoid $A$. The segment connecting the point $x_1$ on the $x$-axis and the point $y_2$ on the $y$-axis is called a \textit{diagonal} of $A$. Similarly, the segment connecting the point $x_2$ on the $x$-axis and the point $y_1$ on the $y$-axis is another diagonal of $A$. A trapezoid has at most two diagonals and at least one diagonal (when $x_1 = x_2$ and $y_1 = y_2$). We assume that all corners of each trapezoid in the model have mutually different $x$- and $y$-coordinates. Otherwise, we may obtain this property by perturbing the corner points without changing the relationship between trapezoids \cite{felsner1997trapezoid}. Therefore, we can assume that all $x$- and $y$-coordinates are integers in the interval $[1, 2n]$, where $n = |\tau|$ is the number of trapezoids by mapping each coordinate to an integer in this interval and keeping their cardinality order (see Figure \ref{F:trapezoid box representation} for an example).
	
	We construct an undirected graph $G(\tau) = (V, E)$ by setting the vertex set $V = \{1, 2, ..., |\tau|\}$ and labeling each trapezoid in $\tau$ as a distinct number from $1$ to $|\tau|$. Then, $AB$ is an edge of $G(\tau)$ if and only if trapezoid $A \in \tau$ and trapezoid $B \in \tau$ intersect (i.e. when a diagonal of $A$ intersects with a diagonal of $B$). An undirected graph $G$ is a trapezoid graph if there exists a set $\tau$ such that $G$ is isomorphic to $G(\tau)$.

\section{Maximum induced matching in permutation graphs} \label{S:MIM in permutation graphs}

\subsection{An $\mathcal{O}(n^2)$ maximum induced matching algorithm in permutation graphs} \label{SS:O(n^2) MIM in permutation graphs}

	Permutation $\pi^{-1} = (\pi^{-1}(1), \pi^{-1}(2), ..., \pi^{-1}(n))$ can be represented as points on 2-dimensional \texttt{SPACE}$(\pi^{-1})$ with horizontal axis $i$ and vertical axis $\pi^{-1}$. Each element $\pi^{-1}(i)$ corresponds to the point $(i, \pi^{-1}(i))$ on \texttt{SPACE}$(\pi^{-1})$. An edge $xy$ of $G(\pi)$ (or equivalently, vertex $xy$ of $L(G)^2$) is described as a rectangle whose sides are parallel to the axes and having two opposite corners $(x, \pi^{-1}(x))$ and $(y, \pi^{-1}(y))$ (see Figure \ref{F:permutation box representation} for an example). The problem could be viewed from a different angle as finding a longest sequence of disjoint rectangles such that the next rectangle is completely at the top-right of the previous rectangle in the sequence, since such a sequence corresponds to a maximum independent set in $L(G)^2$ and vice versa. This geometric representation reveals special benefits based on the \textit{sweep line technique}This geometric representation reveals special benefits based on the sweep line technique that plays an essential role in our algorithms. A sweep line moving from right to left on \texttt{SPACE}$(\pi^{-1})$ determines the order of dynamic programming processes, which helps to find a MIM correctly and efficiently.

\begin{figure}[H]
\centering
\includegraphics[width=14cm]{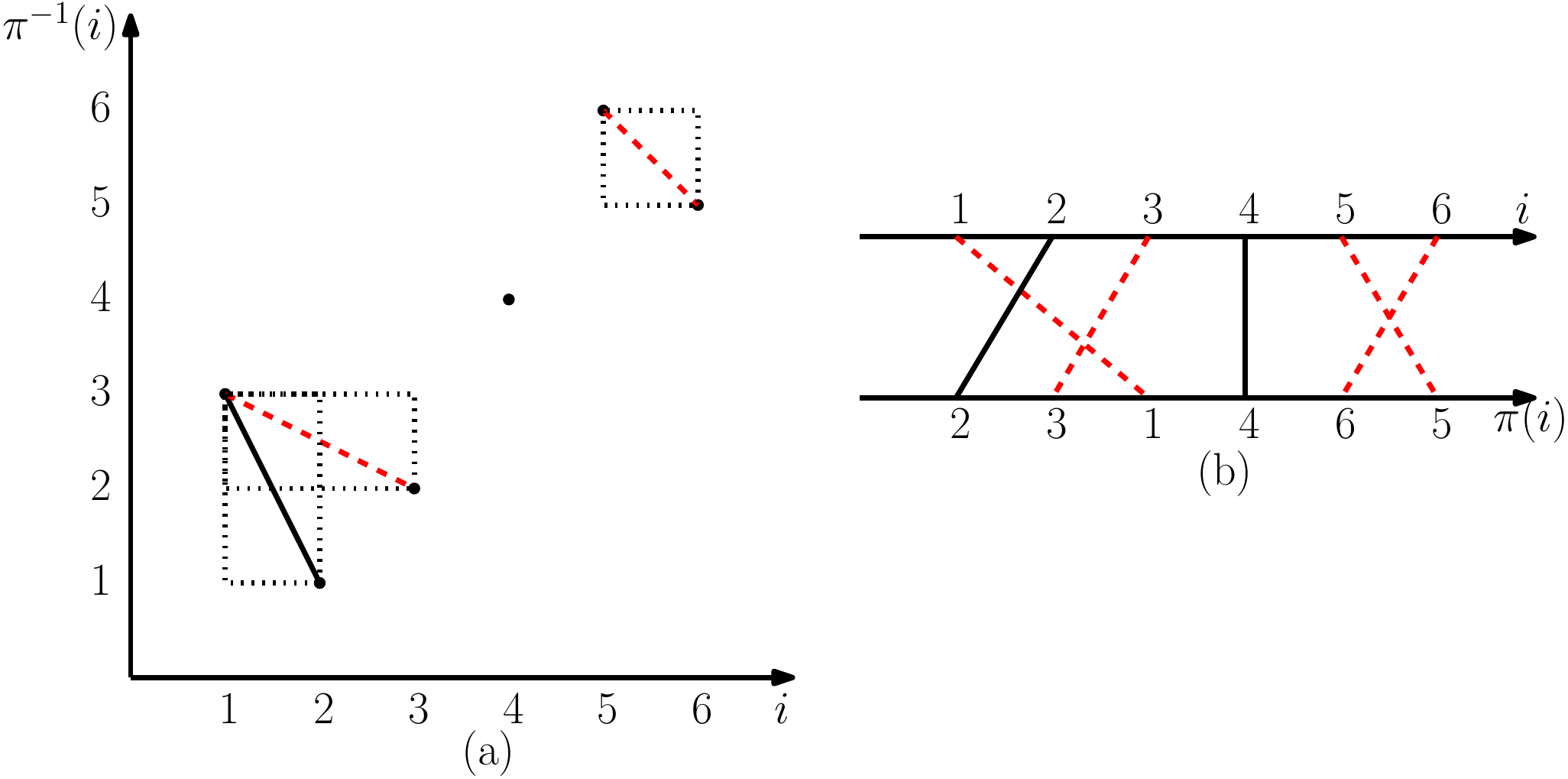}
\caption{\rm The presentation of permutation $\pi^{-1} = (3, 1, 2, 4, 6, 5)$ on \texttt{SPACE}$(\pi^{-1})$ (a) and the corresponding diagram of permutation $\pi$ (b). There are three matches: (1, 2), (1, 3) and (5, 6). Match (5, 6) is greater than the other two matches. A longest chain of length 2 is (1, 3), (5, 6).}
\label{F:permutation box representation}
\end{figure}

	We first briefly describe an algorithm that finds a MIM on $G(\pi)$ in $\mathcal{O}(n^2)$ time. Let $G(\pi) = (V, E)$ where $\pi$ is a permutation of length $n$ and $E$ is the edge set of size $m$. We show some following definitions.

\begin{definition}
	An ordered pair $(x, y)$ is called a \textit{match} if $1 \leq x < y \leq n$ and $\pi^{-1}(x) > \pi^{-1}(y)$. For a match $(x, y)$, $x$ is called the \textit{left end} and $y$ is called the \textit{right end} of the match.
\end{definition}

	One can see that each match $(x, y)$ corresponds to the edge $xy$ of $G(\pi)$, so the number of distinct matches in \texttt{SPACE}$(\pi^{-1})$ is $m$. Denote by $E'$ the set of all matches in \texttt{SPACE}$(\pi^{-1})$. In our algorithms, the term \textit{match} is used instead of \textit{edge}.  The utilization of \textit{match} in place of \textit{edge} is necessary to help our sweep line correctly determine the order of the dynamic programming process lately described.

\begin{definition}
	Given two matches $e = (x, y)$ and $e' = (x', y')$, we define $e < e'$ if $y < x'$ and $\pi^{-1}(x) < \pi^{-1}(y')$. We say that $e$ is \textit{smaller} than $e'$ and $e'$ is \textit{greater} than $e$.
\end{definition}

\begin{definition}
	A sequence of matches $e_1, e_2, ..., e_k$ is called a \textit{chain} if $e_i < e_{i+1}$ for all $1 \leq i < k$. The \textit{length} of the chain is $k$, and the match $e_1$ is the \textit{smallest match} of the chain.
\end{definition}

	Since a chain corresponds to a maximum independent set in $L(G)^2$ and vice versa, the MIM problem in permutation graphs turns out to be finding the longest chain on \texttt{SPACE}$(\pi^{-1})$. It could be solved by computing a function $f: E' \rightarrow \mathbb{N}$, where $f(e)$ is the length of the longest chain having $e$ as the smallest match. One can see that the maximum value of $f(e)$ among all matches $e$ is the size of a MIM on $G(\pi)$, which is what we need. We also compute a function $link: E' \rightarrow E'$, where $link(e)$ is a match such that $f(e) = f(link(e)) + 1$ and the match $e$ is smaller than $link(e)$. This $link$ function is utilized to construct a MIM in the end.

	\smallskip
	The whole algorithm can be summarized as follow:

\begin{framed}
	\textbf{Steps to find a MIM in $G(\pi)$:}

	1. Construct all matches that exist on \texttt{SPACE}$(\pi^{-1})$ and store them in adjacent lists.
	
	2. Calculate two functions $f$ and $link$ for all matches $e$.
	
	3. Build a MIM based on functions $f$ and $link$ calculated in step 2.
\end{framed}

\subsubsection{Construct all matches that exist on \texttt{SPACE}$(\pi^{-1})$} \label{SSS:permutation step construct}
	
	Throughout this paper, we denote $A.x$ as the value $x$ of an object $A$.
	
	In the first step, we construct all matches that exist on \texttt{SPACE}$(\pi^{-1})$. Let $Match(x)$, where $1 \leq x \leq n$, be the list of all $y\ (1 \leq y < x)$ having $\pi^{-1}(y) > \pi^{-1}(x)$. The set $\{(y, x)\ |\ y \in Match(x)\}$ is the set of all matches with the right end $x$. These sets are pairwise disjoint for different values of $x$ and the union of them is all the matches that exist on \texttt{SPACE}$(\pi^{-1})$. $Match$ lists can be considered as adjacent lists of the permutation graph, except that each edge is stored in only one list.
	
	We use the linked list data structure to construct all matches that exist on $\mathtt{SPACE}(\pi^{-1})$ in $\mathcal{O}(m + n)$ time (see Procedure \ref{Algo:buildAllMatches-Permutation} below). We create a linked list \texttt{LL} of $n$ nodes numbered from 1 to $n$, where node $n$ is the head of \texttt{LL}, and the next node of node $i$ is node $(i - 1)$ for all $1 < i \le n$. Initially, we set $Match(x) = \emptyset$ for all $x$. We make a loop from $a = n$ to $a = 1$. For each $a$, we start the visiting process from the head of \texttt{LL}. When a node $p$ is visited, if $p > \pi(a)$, then we add $\pi(a)$ to $Match(p)$ and move to the next node, else we remove $p$ from \texttt{LL}.
	
	For readability and simplicity, we only show a brief version of Procedure \ref{Algo:buildAllMatches-Permutation} here. For the complete version of Procedure \ref{Algo:buildAllMatches-Permutation}, please see \textit{Appendix A}.
	
	\begin{algorithm}[h]\small
\caption{$\mathtt{buildAllMatches}(\pi)$}
\label{Algo:buildAllMatches-Permutation}

\SetKwInOut{Description}{Description}
\Description{\textit{Step 1:} Construct all matches that exist on \texttt{SPACE}$(\pi^{-1})$ and store them as adjacent lists.}

\ShowLn \tcc{initialize}
Create linked list \texttt{LL} of $n$ nodes numbered from 1 to $n$, where $\mathtt{LL}.head = n$, and $i.next = i - 1$ for all $1 < i \le n$\;
$Match(x) \leftarrow \emptyset$ for all $1 \leq x \leq n$\;
\;
\ShowLn \tcc{build $Match$ lists}
\For{$a \leftarrow n$ down to $1$}{
	$p \leftarrow \mathtt{LL}.head$\;
	\While{$p > \pi(a)$}{
		add $\pi(a)$ to $Match(p)$\;
		$p \leftarrow p.next$\;
	}
	remove node $p$ from $\mathtt{LL}$\;
}
\;
\Return {$\{Match(x)\ |\ 1 \leq x \leq n\}$}\;

\end{algorithm}

\begin{lemma} \label{LM:buildAllMatches correctness}
	Procedure \texttt{buildAllMatches} correctly constructs the list $Match(x)$ for all $1 \le x \le n$.
\end{lemma}

\begin{proof}
Because the nodes in the linked list \texttt{LL} are sorted in decreasing order from head to tail and there are $n$ distinct nodes at the beginning, the while-loop on line 8 always terminates when $p = \pi(a)$. Therefore, when $a = i$ for some $1 \le i \le n$, only the nodes $\pi(n)$, $\pi(n - 1)$, ..., $\pi(i + 1)$ are removed from \texttt{LL}. It means that when $a = i$, we have $\pi^{-1}(\pi(a)) = a = i > \pi^{-1}(p)$ for all $p > \pi(a)$. Consequently, all the elements added to the list $Match(p)$ on line 9 are valid.

On the other hand, suppose that when $a = i$ for some $1 \le i \le n$, there is some removed node $q$ where $(\pi(a), q)$ is a match, we will show a contradiction. Indeed, if $q$ is removed before, $q$ must belong to the set $\{\pi(n), \pi(n - 1), ..., \pi(i + 1)\}$. It leads to the fact that $\pi^{-1}(q) \ge i + 1 > i = a = \pi^{-1}(\pi(a))$, so $(\pi(a), q)$ is not a match, contradicts with the assumption. Hence, when $a = i$, all the nodes $p$ where $(\pi(a), p)$ is a match still remain in the linked list \texttt{LL}. The while-loop on line 8 can iterate through all such nodes since $p > \pi(a)$ when $(\pi(a), p)$ is a match as definition.

Based on these two conclusions, for all $1 \le x \le n$, Procedure \textit{buildAllMatches} correctly constructs the list $Match(x)$.
\end{proof}

\begin{lemma} \label{LM:buildAllMatches running time}
	Procedure \texttt{buildAllMatches} takes $\mathcal{O}(m + n)$ time.
\end{lemma}

\begin{proof}
	The time complexity of procedure \texttt{buildAllMatches} is the \texttt{LL}'s building time plus the number of times we add a new element to a $Match$ list. Since each $Match$ list does not have duplicate elements and each match corresponds to a unique edge in $G(\pi)$, this procedure takes $\mathcal{O}(m + n)$ time.
\end{proof}

\subsubsection{Calculate two functions $f$ and $link$} \label{SSS:permutation step calculate}

	In the second step, we calculate two functions $f$ and $link$. We make a sweep line $L$ moving from right to left and visit every coordinate $i = x\ (1 \leq x \leq n)$. On $L$, we maintain $n$ memory units called \textit{cells}. Each time when $L$ stays at $i = x$, each cell $L_y$ ($1 \leq y \leq n$) is a pair $(len, trace)$, where $len$ is the length of a longest chain where the smallest match is any match $trace = (z, t)$ satisfying $z > x$ and $\pi^{-1}(t) = y$. An algorithm to calculate the function $f$ can be built with the aid of the sweep line $L$. First of all, let $S_x$ be a list, which stores pairs of match $e$ and its corresponding $f(e)$, for all matches $e$ having $x$ as its left end. The $S_x$ lists are needed later for the cell updating process in our dynamic algorithm. Initially, we set $S_x = \emptyset$, also set all $L_y$ to be $(0, \mathtt{NULL})$. Then, we start to move the sweep line $L$. When reaching the coordinate $i = x$, we calculate $f(e)$ and $link(e)$ for all matches $e = (y, x)$ where $y \in Match(x)$ by the following formulae:

\begin{framed}
	$\bullet\ f(e) = 1 + \underset{j > \pi^{-1}(y)}{\max} L_j.len$
	
	$\bullet\ link(e) = \underset{j > \pi^{-1}(y), f(e) = L_j.len + 1\ (*)}{L_j.trace}$ (for an arbitrary $j$ satisfies (*))
\end{framed}
	
	After having $f(e)$, we add $(f(e), e)$ into the list $S_y$. After the calculations of $f(e)$ and $link(e)$ for all $e$ having $x$ as their right end, we start to update cells on $L$. This process is done by going through all elements in the list $S_x$: for each element $(f(e), e)$ where $e = (x, a)$, if $L_{\pi^{-1}(a)}.len < f(e)$, we set $L_{\pi^{-1}(a)} = (f(e), e)$ (see Figure \ref{F:O(n^2) permutation G(pi)}b for an example). In the end, when all $x$ are swept by $L$, we will have the answer for the MIM problem by looking up the functions $f$ and $link$.
	
	A naive algorithm based on this method takes $\mathcal{O}(mn)$ time to run. Indeed, to calculate $f(e)$ and $link(e)$ for each $e$, a single loop that runs in $\mathcal{O}(n)$ is required. As the number of matches is $m$, it takes $\mathcal{O}(mn)$ overall. The main problem which causes the algorithm slow is the requirement of an $\mathcal{O}(n)$ loop to calculate each match.
	
	However, it is noticeable that the elements in each list $Match(x)$ are arranged in decreasing order of the function $\pi^{-1}$ applying to them, \textit{i.e.}, element $a$ is added before element $b$ in $Match(x)$ if $\pi^{-1}(a) > \pi^{-1}(b)$. The following lemma will prove this argument.

\begin{lemma} \label{LM:Match(x) elements order}
	The elements of each list $Match(x)$ are arranged in decreasing order of the function $\pi^{-1}$ applying to them.
\end{lemma}

\begin{proof}
	Consider any pair $(p, q)$ of elements in $Match(x)$, suppose that $p = \pi(i)$ is added before $q = \pi(j)$, we can see that $i > j$ since the for-loop iterator $a$ in \texttt{buildAllMatches} decreases. Therefore, $\pi^{-1}(p) = \pi^{-1}(\pi(i)) = i > j = \pi^{-1}(\pi(j)) = \pi^{-1}(q)$, the lemma is proven.
\end{proof}
	
	By Lemma \ref{LM:Match(x) elements order}, we can reduce running time by using just $\mathcal{O}(n)$ operations to calculate $f(e)$ and $link(e)$ for all $e$ with the right end $x$. Hence, we have an $\mathcal{O}(n^2)$ algorithm overall. We maintain a decreasing-pointer $z$ and two variables $maxLen$ and $trace$ where $maxLen$ is $\underset{i > z}{\max}(L_i.len)$ and $trace$ is the corresponding $L_i.trace$ when $L_i.len$ reaches maximum. These two variables $maxLen$ and $trace$ are updated each time we decrease $z$ (see Figure \ref{F:O(n^2) permutation G(pi)}a for an example). The following procedure \texttt{calculateFAndLink} will represent step 2. For readability and simplicity, we only show a brief version of Procedure \ref{Algo:calculateFAndLink-Permutation} here. For the complete version of Procedure \ref{Algo:calculateFAndLink-Permutation}, please see \textit{Appendix B}. \medskip

\begin{algorithm}[H]\small
\caption{$\mathtt{calculateFAndLink}(\pi^{-1}, Match)$}
\label{Algo:calculateFAndLink-Permutation}

\SetKwInOut{Description}{Description}
\Description{\textit{Step 2:} Calculate two functions $f$ and $link$ for all matches $e$.}

\ShowLn \tcc{initialize}	
initialize $S$ and $L$\;
\;
\ShowLn \tcc{calculate functions $f$ and $link$}
\For{$x \leftarrow n$ down to $1$}{
	calculate $f$ and $link$ for all matches $(*, x)$ with pointer $z$ decreasing from $n$\;
	update elements in $S_x$ to $L$\;
}
\;
\Return {$(f, link)$}\;

\end{algorithm}

\begin{figure}[H]
\vspace{2mm}
\centering
\includegraphics[width=13cm]{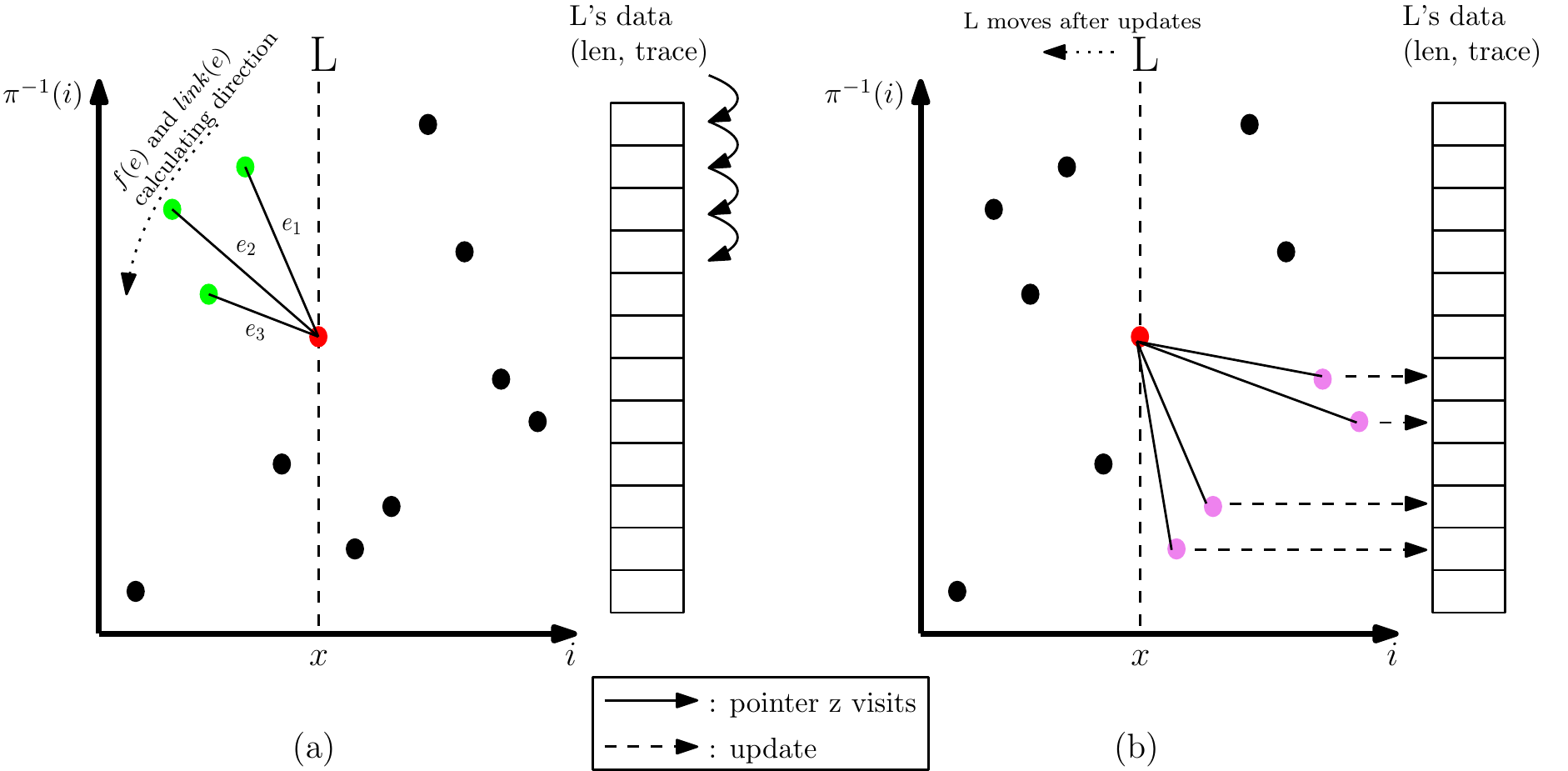}
\caption{Calculate $f$ and $link$ in $G(\pi)$ in $\mathcal{O}(n^2)$ time.\\
(a) When $L$ stays at $i = x$, for all matches $e = (y, x)$, $f(e)$ and $link(e)$ are calculated in decreasing order of $\pi^{-1}(y)$. To obtain $f(e)$ and $link(e)$, every $L_z$ such that $\pi^{-1}(y) < z \leq n$ is visited. Then $(f(e), e)$ is added to the list $S_y$.\\
(b) Each element $(f(e), e)$ having $e = (x, y)$ in $S_x$ is updated to the cell $L_{\pi^{-1}(y)}$. After all, the sweep line $L$ moves to the coordinate $i = x - 1$.}
\label{F:O(n^2) permutation G(pi)}
\end{figure}

\subsubsection{Build a maximum induced matching} \label{SSS:permutation step build}

	Let $startChain$ be a match where $f(startChain)$ is the maximum among all $f(e)$. Consequently, a MIM on $G(\pi)$ will have the cardinality of $f(startChain)$. Such a MIM could be built by tracing the $link$ function in $\mathcal{O}(|\text{MIM}|)$ which is $\mathcal{O}(n)$. Step 3 is implemented in Procedure \texttt{buildMIM} below.

\begin{algorithm}[h]\small
\caption{$\mathtt{buildMIM}(f, link)$}

\SetKwInOut{Description}{Description}
\Description{\textit{Step 3:} Build a MIM based on functions $f$ and $link$ calculated in step 2.}

\ShowLn \tcc{initialize}
$\mathtt{MIM} \leftarrow \emptyset$\;
$startChain \leftarrow a$ (for an arbitrary $a$ where $f(a)$ is maximum)\;
$e \leftarrow startChain$\;
\;
\ShowLn \tcc{build a MIM}
\While{$e \neq \mathtt{NULL}$}{
	$(x, y) \leftarrow e$\;
	$\mathtt{MIM} \leftarrow \mathtt{MIM} \cup xy$\;
	$e \leftarrow link(e)$\;
}
\;
\Return {$\mathtt{MIM}$}\;

\end{algorithm}

\subsubsection{Summary} \label{SSS:permutation summary}

    Given a permutation model $\pi$, a MIM of permutation graph $G(\pi)$ can be found by calling the following procedure \texttt{maxInducedMatching}.

\begin{algorithm}[h]\small
\caption{$\mathtt{maxInducedMatching}(\pi)$}
\label{Algo:maxInducedMatching-Permutation}

\SetKwInOut{Description}{Description}
\Description{Finding a MIM in $G(\pi)$}

$Match \leftarrow \mathtt{buildAllMatches}(\pi)$\;
$(f, link) \leftarrow \mathtt{calculateFAndLink}(\pi^{-1}, Match)$\;
\Return {$\mathtt{buildMIM}(f, link)$}\;
\end{algorithm}

\begin{theorem} \label{T:O(n^2) permutation G(pi)}
	A maximum induced matching in permutation graph $G(\pi)$ can be found in $\mathcal{O}(n^2)$ time.
\end{theorem}

\begin{proof}
    Procedure \ref{Algo:maxInducedMatching-Permutation} returns a MIM of permutation graph $G(\pi)$. The first function \textit{buildAllMatches} takes $\mathcal{O}(m + n)$ time, proved by Lemma \ref{LM:buildAllMatches running time}. The second function \textit{calculateFAndLink} runs in $\mathcal{O}(n^2)$, as shown in Section \ref{SSS:permutation step calculate}. The last function \textit{buildMIM} takes $\mathcal{O}(n)$ time to run, as shown in Section \ref{SSS:permutation step build}. Therefore, Procedure \ref{Algo:maxInducedMatching-Permutation} a maximum induced matching on $G(\pi)$ in $\mathcal{O}(n^2)$ time.
\end{proof}

\subsection{Faster maximum induced matching algorithms in permutation graphs} \label{SS:O(mlog n + n) and O(mloglog n + n) MIM in permutation graphs}

	With the aid of a segment tree \cite{bentley1977algorithms}, we can build an $\mathcal{O}(m\log n + n)$ algorithm for MIM in permutation graphs from the $\mathcal{O}(n^2)$ algorithm. Unlike the $\mathcal{O}(n^2)$ algorithm in which the sweep line $L$ stores an array of cells, $L$ here stores a segment tree. All operations, including updating a cell and finding the maximum cell within an interval, are done in $\mathcal{O}(\log n)$ time per each operation. As we mentioned, a MIM on permutation graph $G$ can be seen as a maximum independent set on trapezoid graph $L(G)^2$. Actually, this $\mathcal{O}(m\log n + n)$ algorithm is similar to the maximum independent set algorithm for trapezoid graph in \cite{felsner1997trapezoid}, which could be improved to an $\mathcal{O}(m\log\log n + n)$ solution by using vEB tree in \cite{van1977preserving}. A Van Emde Boas tree (or Van Emde Boas priority queue), also called vEB tree, supports searching, inserting and deleting an element in $\mathcal{O}(\log\log M)$, where $M = 2^m$ is a fixed number indicating the maximum number of nodes to be stored in the tree, and the elements are integers in $\{1, 2, ..., M\}$. A vEB tree implements an associative array of $m$-bit integer keys in $\mathcal{O}(M)$ space.
	
	We shall not show this algorithm in detail since our following $\mathcal{O}(m + n)$ algorithm does not use the idea of segment tree nor vEB tree.
	
\subsection{An $\mathcal{O}(m + n)$ maximum induced matching algorithm in permutation graphs} \label{SS:O(m + n) MIM in permutation graphs}
	
	To the best of our knowledge, there has not been existed any data structure that supports both query and update operations in $\mathcal{O}(1)$ applicable for this problem. Therefore, improving an $\mathcal{O}(n^2)$ algorithm in Section \ref{SS:O(n^2) MIM in permutation graphs} into an $\mathcal{O}(m + n)$ algorithm only by applying different data structures in Step 2 is quite an impossible work.
	
	Let pay attention to procedure \texttt{maxInducedMatching}. Although the idea of pointer $z$ is essential to make the running time $\mathcal{O}(n^2)$, it consumes $\mathcal{O}(n)$ calculations for each $x$ from $n$ to $1$, and becomes the most time-consuming part of the whole algorithm. Instead of $\mathcal{O}(n)$, if we can turn it to $\mathcal{O}(|Match(x)|)$ for each $x$, we will acquire an $\mathcal{O}(m + n)$ for the overall algorithm. We show next the most critical points for an $\mathcal{O}(m + n)$ solution.

	Assume that $L_y$ now stores 3 values $(len, trace, swept)$ in place of a pair $(len, trace)$, where $swept$ equals to $1$ if the sweep line $L$ has passed the coordinate $i = \pi(y)$ ($L_y$ is called a \textit{swept cell}) or equal to $0$ ($L_y$ is called an \textit{unswept cell}) otherwise. Some operations of our $\mathcal{O}(n^2)$ algorithm are going to be changed.

\subsubsection{Adjustment of formulae} \label{SSS:Adjustment of formulae}
	
	We call $\varphi(a)$ the greatest number smaller than $a$ having $L_{\varphi(a)}.swept = 0$ (if such $\varphi(a)$ does not exist then we assume $\varphi(a) = 0$). To calculate $f(e)$ and $link(e)$ for all matches $e = (y, x)$ where $y \in Match(x)$, their formulae are also adjusted as below:
	
\begin{framed}
	$\bullet\ f(e) = 1 + \underset{j \geq \pi^{-1}(y),\ L_j.swept = 0}{\max} L_j.len$
	
	$\bullet\ link(e) = \underset{j \geq \pi^{-1}(y),\ L_j.swept = 0,\ f(e) = L_j.len + 1\ (**)}{L_j.trace}$ (for an arbitrary $j$ satisfies (**))
\end{framed}

    With the new formulae, we do not need the pointer $z$ decreasing from $n$ anymore. We just need to iterate through all $L_{\pi^{-1}(z)}$ where $n \geq \pi^{-1}(z) > \pi^{-1}(x)$ and $L_{\pi^{-1}(z)}.swept = 0$ (see Figure \ref{F:O(m + n) permutation G(pi)}a for an example). In this occasion, $L_{\pi^{-1}(z)}.swept = 0$ also means $z < x$, because $L_{\pi^{-1}(z')}.swept$ is set to $1$ for all $z' > x$ when sweep line $L$ stays at coordinate $i = x$. Amazingly, if we have $\pi^{-1}(z) > \pi^{-1}(x)$ and $z < x$, then $(z, x)$ is a match. Therefore, the number of $L_{\pi^{-1}(z)}$ we need to check for each $x$ is exactly $\mathcal{O}(|Match(x)|)$ as we need.

    When going through all elements of the list $S_x$ in the updating process, for each element $(f(e), e)$ where $e = (x, a)$, we will update $L_{\varphi(\pi^{-1}(a))}$ instead of $L_{\pi^{-1}(a)}$. If $L_{\varphi(\pi^{-1}(a))}.len < f(e)$ then we set $L_{\varphi(\pi^{-1}(a))}.len = f(e)$ and $L_{\varphi(\pi^{-1}(a))}.trace = e$ (see Figure \ref{F:O(m + n) permutation G(pi)}b for an example).

	Finally, before moving the sweep line $L$ to the coordinate $i = x - 1$, we set $L_{\pi^{-1}(x)}.swept = 1$, $L_{\pi^{-1}(x)}$ becomes a swept cell. In addition, if $L_{\pi^{-1}(x)}.len > L_{\varphi(\pi^{-1}(x))}.len$, then we set $L_{\varphi(\pi^{-1}(x))} = (L_{\pi^{-1}(x)}.len, L_{\pi^{-1}(x)}.trace, 0)$ (see Figure \ref{F:O(m + n) permutation G(pi)}b for an example).
	
\begin{figure}[h]
\vspace{1mm}
\centering
\includegraphics[width=14cm]{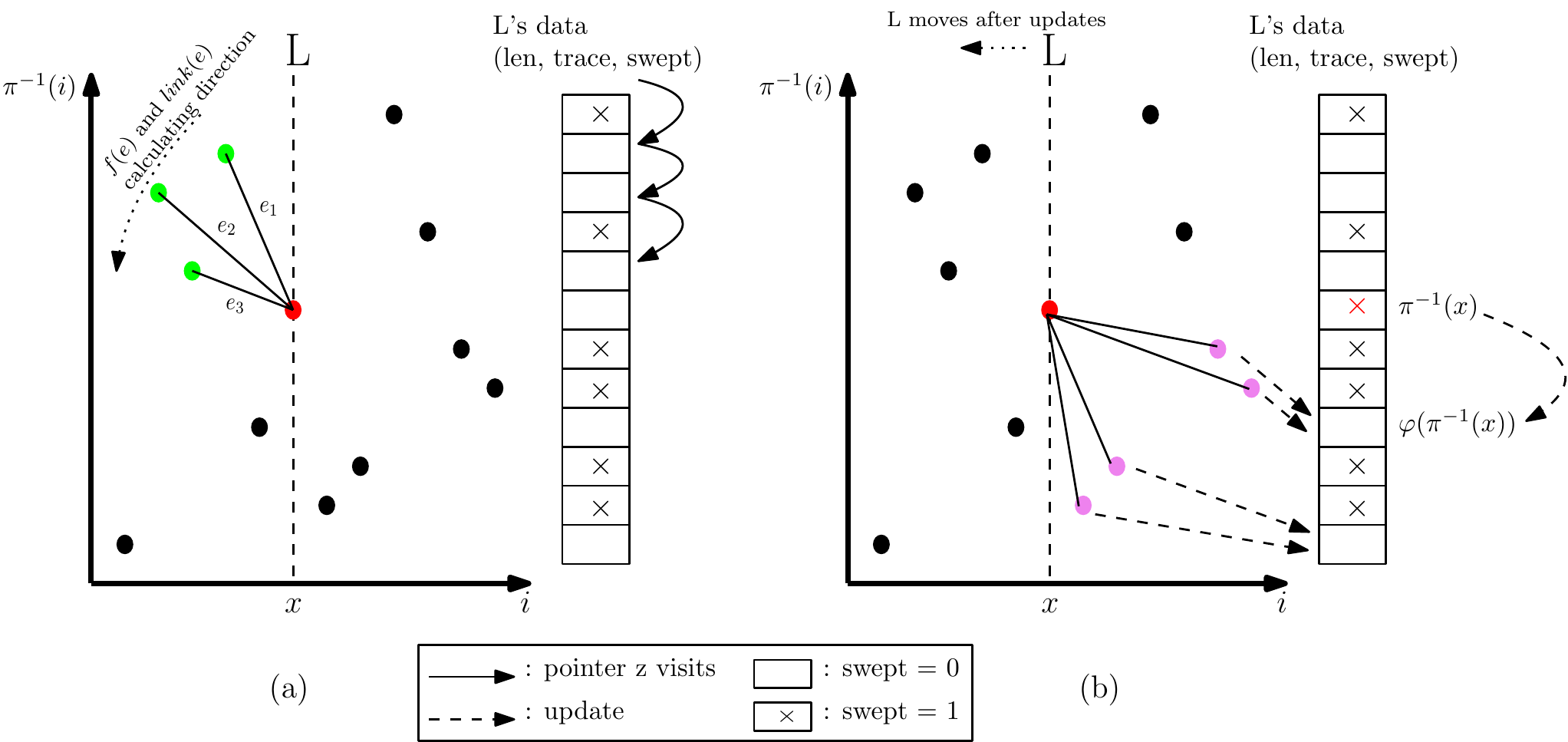}
\caption{Calculate $f$ and $link$ in $G(\pi)$ in $\mathcal{O}(m + n)$ time.\\
(a) When $L$ stays at $i = x$, for all matches $e = (y, x)$, $f(e)$ and $link(e)$ are calculated in decreasing order of $\pi^{-1}(y)$. To obtain $f(e)$ and $link(e)$, every $L_z$ such that $\pi^{-1}(y) \leq z \leq n$ and $L_z.swept = 0$ is visited. Then $(f(e), e)$ is added to the list $S_y$.\\
(b) Every element $(f(e), e)$ having $e = (x, y)$ in $S_x$ is updated to the cell $L_{\varphi(\pi^{-1}(y))}$. Then $L_{\pi^{-1}(x)}.swept$ is set to $1$; $L_{\pi^{-1}(x)}.len$ and $L_{\pi^{-1}(x)}.trace$ are updated to $L_{\varphi(\pi^{-1}(x))}$. After all, the sweep line $L$ moves to the coordinate $i = x - 1$.}
\label{F:O(m + n) permutation G(pi)}\vspace{1mm}
\end{figure}
	
	Let $\psi(y)$ be the smallest number such that $n \geq \psi(y) > y$ and $L_{\psi(y)}.swept = 0$, if such $\psi(y)$ does not exist then we assume $\psi(y) = n$. We will prove the correctness of all these changes above through Lemma \ref{LM:permutation push down lemma}.

\begin{lemma} \label{LM:permutation push down lemma}
	For all $1 \leq x \leq n$, when the sweep line $L$ moves from coordinate $i = x + 1$ to coordinate $i = x$, each cell $L_y\ (1 \leq y \leq n)$ where $L_y.swept = 0$ contains 3 variables $(len, trace, swept)$ that:
\begin{itemize}
\itemsep=0.96pt
	\item $L_y.len$ is the length of a longest chain among all chains having their smallest match $(a, b)$ satisfies $a > x$ and $y < \pi^{-1}(b) \leq \psi(y)$.
	\item $L_y.trace$ is the smallest match of such a chain (if there are multiple choices, $L_y.trace$ can be any match).
	\item $L_y.swept = 0$.
\end{itemize}
\end{lemma}

\begin{proof}
    We will prove this lemma using the induction hypothesis as follows.\vspace*{-1mm}
    \begin{itemize}
    \itemsep=0.8pt
        \item The lemma is true for $x = n$ since for all $1 \leq y \leq n$, $L_y = (0, \mathtt{NULL}, 0)$.
        \item Suppose that the lemma is true for some $x = t + 1$ where $1 \leq t < n$, we will prove that the lemma is correct for $x = t$. Indeed, suppose that $\pi^{-1}(t + 1) = k$, we will have $\varphi(\psi(k)) = \varphi(k)$ and $\psi(\varphi(k)) = \psi(k)$ after all update operations are done and before the sweep line $L$ moves from coordinate $i = t + 1$ to coordinate $i = t$. In addition, since we set $L_{\varphi(k)} = (L_{k}.len, L_{k}.trace, 0)$ if $L_{k}.len > L_{\varphi(k)}.len$ before $L$'s movement, the cell $L_{\varphi(k)}$ will satisfy all the three properties mentioned in Lemma \ref{LM:permutation push down lemma}. Therefore, each cell $L_y\ (1 \leq y \leq n)$ where $L_y.swept = 0$ will also satisfy all these three properties.
    \end{itemize}\vspace*{-1mm}
    By the induction hypothesis, we can conclude that the lemma applies for all $1 \leq x \leq n$.
\end{proof}
	
	The last issue is whether there is a data structure, which helps to calculate all the $\varphi$ operations in $\mathcal{O}(m + n)$. Here the disjoint-set data structure does the job.\vspace*{-2mm}
	
\subsubsection{Complexity improvement by the disjoint-set data structure}  \label{SSS:disjoint-set}
	
	We consider a disjoint-set data structure $d$ consists of $n$ sets $\{1\}, \{2\}, ..., \{n\}$, where the \textit{name of the set} containing $i$ is also $i$ at the beginning. Each set in the disjoint-set structure corresponds to either a single unswept cell or a set of consecutive swept cells on the sweep line $L$. The purpose of $d$ is to quickly jump through sets of consecutive cells that are swept, thus lead to quickly calculate the function $\varphi$.
	
\medskip
	We define two operations on our disjoint-set data structure $d$ as below:\vspace*{-1mm}
		\begin{itemize}
\itemsep=0.8pt
	    \item $\mathtt{find}(d, x)$: Return the name of the set containing $x$.
	    \item $\mathtt{union}(d, x, y)$: Create a new set that is the union of the sets containing $x$ and $y$. The name of the new set is the name of the old set containing $x$. This operation assumes that $x$ and $y$ are initially in different sets and destroys the old sets containing $x$ and $y$.
	\end{itemize}\vspace*{-1mm}
	
	In our algorithm, before moving the sweep line $L$ from the coordinate $i = x$ to the coordinate $i = x - 1$, we set $L_{\pi^{-1}(x)}.swept = 1$. After this operation, if $L_{\pi^{-1}(x)-1}.swept = 1$, then the set containing $\pi^{-1}(x)$ will be united with the set containing $\pi^{-1}(x) - 1$ by calling $\mathtt{union}(d, \pi^{-1}(x) - 1, \pi^{-1}(x))$. Similarly if $L_{\pi^{-1}(x)+1}.swept = 1$, we unite two sets where $\pi^{-1}(x)$ and $\pi^{-1}(x) + 1$ belong to by calling $\mathtt{union}(d, \pi^{-1}(x), \pi^{-1}(x) + 1)$.
	
	Since these are the only $\mathtt{union}$ operations in our algorithm, we can see that the name of the set containing $x$ is always the minimum number in that set. Based on this observation, we have a simple way to calculate $\varphi$ as follows. If $L_{x-1}.swept = 0$, then $\varphi(x) = x - 1$. Otherwise, $\varphi(x)$ is exactly $\mathtt{find}(d, x - 1) - 1$. This is correct since $\mathtt{find}(d, x - 1)$ is the smallest number in the set containing $(x - 1)$. In general, the function $\varphi(x)$ is calculated by the procedure $cal\varphi$.
\eject

\begin{algorithm}[H]\small
\caption{$\mathtt{cal\varphi}(L, d, x)$}

\SetKwInOut{Description}{Description}
\Description{Calculate $\varphi(x)$}

\If{$x = 1$ or $L_{x-1}.swept = 0$}{
	\Return {$x - 1$}\;
}
\Else{
	\Return {$\mathtt{find}(d, x - 1) - 1$}\;
}
\end{algorithm}\medskip

    Thanks to the disjoint-set structure in \cite{gabow1985linear}, all $\mathtt{find}$ and $\mathtt{union}$ operations in our algorithm can be calculated in $\mathcal{O}(m + n)$. Because the total number of $union$ and $find$ operations in Procedure is $\Theta(m + n)$, and according to \cite{gabow1985linear}, in our case, the \textit{union tree} $T$ is the tree where node $i$ is the parent of node $(i + 1)$ for all $i: 1 \leq i < n$. Therefore, the $\mathcal{O}(m + n)$ algorithm can be built by a modification of the procedure \texttt{calculateFAndLink}, shown in Procedure \ref{Algo:calculateFAndLink_m+n-Permutation} below. For readability and simplicity, we only show a brief version of Procedure \ref{Algo:calculateFAndLink_m+n-Permutation} here. For the complete version of Procedure \ref{Algo:calculateFAndLink_m+n-Permutation}, please see \textit{Appendix C}.\smallskip

\begin{algorithm}[H]\small
\caption{$\mathtt{calculateFAndLink}(\pi^{-1}, Match)$}
\label{Algo:calculateFAndLink_m+n-Permutation}

\SetKwInOut{Description}{Description}
\Description{\textit{Step 2:} Calculate two functions $f$ and $link$ for all matches $e$.}

\ShowLn \tcc{initialize}
initialize $S$, $L$ and $d$\;
\;
\ShowLn \tcc{calculate functions $f$ and $link$}
\For{$x \leftarrow n$ down to $1$}{
	calculate $f$ and $link$ for all matches $(*, x)$ with pointer $z$ decreasing from $\varphi(n+1)$\;
	update elements in $S_x$ to $L$\;
	update $L$ and $d$ before $L$ passes coordinate $i = x$\;
}
\;
\Return {$(f, link)$} \;
\end{algorithm}

\begin{theorem} \label{T:O(m + n) permutation G(pi)}
	A maximum induced matching in permutation graph $G(\pi)$ can be found in $\mathcal{O}(m + n)$ time.
\end{theorem}

\begin{proof}
    All $union$ and $find$ operations take $\mathcal{O}(m + n)$ time as shown. Therefore, the cost for Procedure \ref{Algo:calculateFAndLink_m+n-Permutation} is $\mathcal{O}(m + n)$. By substituting Procedure \ref{Algo:calculateFAndLink_m+n-Permutation} for Procedure \ref{Algo:calculateFAndLink-Permutation}, Procedure \ref{Algo:maxInducedMatching-Permutation} can produce a maximum induced matching for $G(\pi)$ in $\mathcal{O}(m + n)$ time.
\end{proof}
	
	McConnell and Spinrad \cite{mcconnell1999modular} introduced an algorithm to construct a permutation model from a permutation graph in linear time. This leads to a linear-time algorithm for the MIM problem in permutation graphs by first generate a permutation model $\pi$ from permutation graph $G$, and then find a MIM in $G(\pi)$ in $\mathcal{O}(m + n)$ time. Based on this result and Theorem \ref{T:O(m + n) permutation G(pi)}, we conclude the section by the following corollary.
	
\begin{corollary}
    A maximum induced matching in a permutation graph $G$ can be found in linear time.
\end{corollary} \label{T:O(m + n) permutation G}

\section{Maximum induced matching in trapezoid graphs} \label{S:MIM in trapezoid graphs}

	Trapezoid graphs are a superclass of permutation graphs. Algorithms for solving the MIM problem in trapezoid graphs, which will be proposed later in this paper, are pretty similar to those of permutation graphs. Despite that, our procedures in Section \ref{S:MIM in permutation graphs} still need a huge modification in order to be applicable in trapezoid graphs.

\subsection{An $\mathcal{O}(m + n)$ maximum induced matching algorithm in trapezoid graphs} \label{SS:O(m + n) MIM in trapezoid graphs}

	A trapezoid graph can be represented as rectangular boxes on 2-dimensional \texttt{SPACE}$(\tau)$ in which each trapezoid corresponds to a unique box. A trapezoid which is made by two intervals $[x_1, x_2]$ and $[y_1, y_2]$, where $x_1 \le x_2$ and $y_1 \le y_2$, is described as a unique rectangle having bottom-left corner $(x_1, y_1)$, top-right corner $(x_2, y_2)$ and edges parallel to $x$- and $y$-axis. 
	
	An edge $AB$ of $G(\tau)$ (or equivalently, a vertex of $L(G)^2$) is described as a \textit{big rectangle} whose sides are parallel to the axes and having two opposite corners $(min(A.x_1, B.x_1), min(A.y_1, B.y_1))$ and $(max(A.x_2, B.x_2), $ $max(A.y_2, B.y_2))$ (see Figure \ref{F:trapezoid box representation} for an example). The problem could be viewed as finding a longest sequence of the disjoint big rectangles such that the next rectangle is completely at the top-right of the previous rectangle in the sequence, since such a sequence corresponds to a maximum independent set in $L(G)^2$ and vice versa. We shall show some definitions similar to Section \ref{S:MIM in permutation graphs}.

\begin{figure}[H]
\centering
\includegraphics[width=12.9cm]{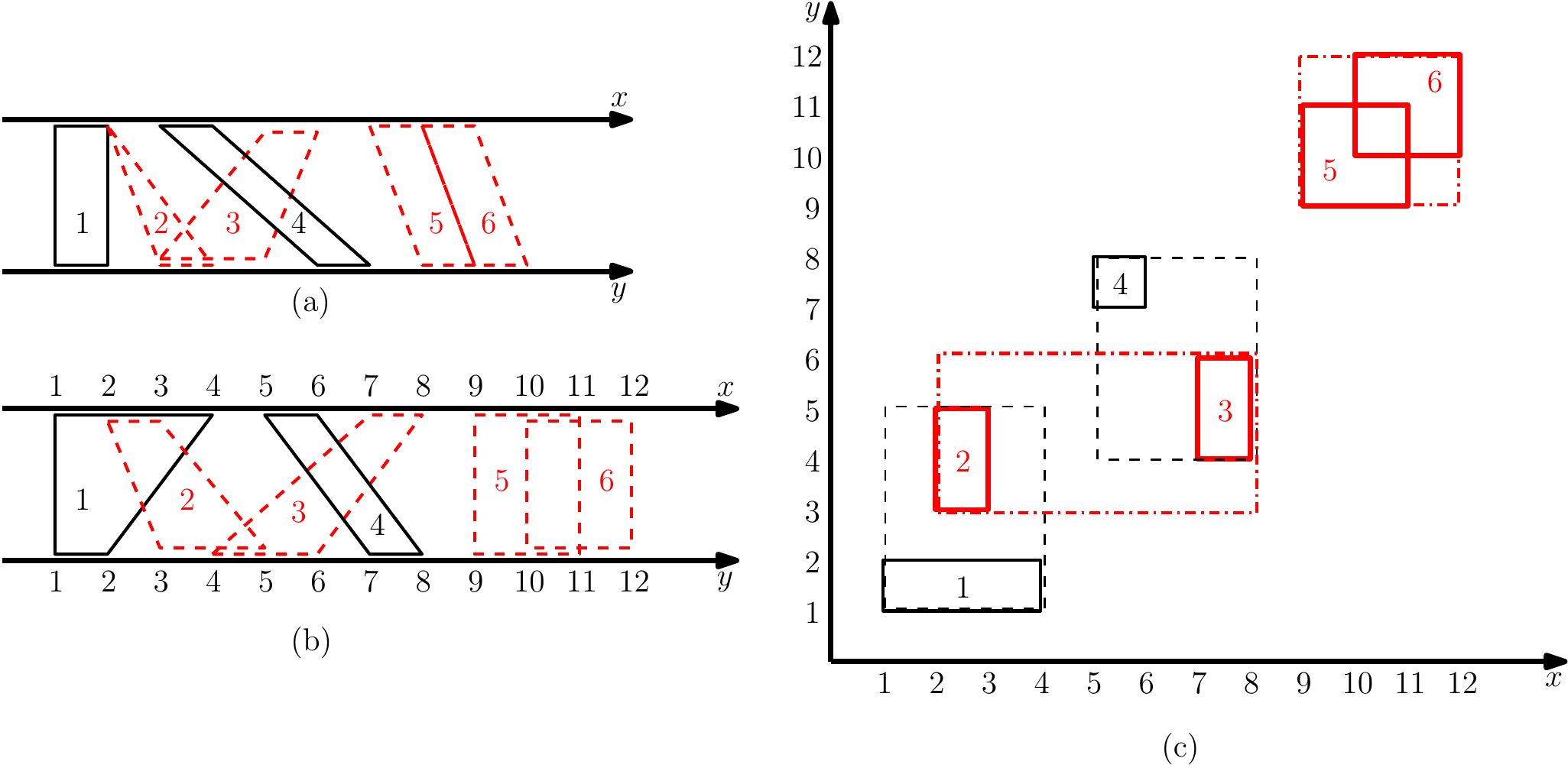}\vspace*{-1mm}
\caption{\rm Trapezoids with duplicate $x$- and $y$-coordinates (a) can be mapped to a new set of trapezoids without duplicate $x$- or $y$-coordinates (b) and can be represented as boxes in 2-dimensional space (c). All $x$- and $y$-coordinates lie inside $[1, 2n]$. There are four matches in this figure: (2, 1), (2, 3), (4, 3) and (5, 6). Match (5, 6) is greater than all other matches. A longest chain of length 2 is (2, 3), (5, 6).}
\label{F:trapezoid box representation}
\end{figure}

\begin{definition}
	An ordered pair of trapezoids $(A, B)$ is called a \textit{match} if $A.x_2 < B.x_2$ and either $A.x_2 \geq B.x_1$ or $A.y_2 \geq B.y_1$. For a match $(A, B)$, $A$ is called the \textit{left end} and $B$ is called the \textit{right end} of the match. The ordered pair $(B, A)$ is called a \textit{reversed match}.
\end{definition}

	One can see that each match $(A, B)$ corresponds to the edge $AB$ of $G(\tau)$, so the number of distinct matches is equal to $m$. In our algorithm, the term \textit{match} is used instead of \textit{edge}.

\begin{definition}
	Given two matches $e\! =\! (A, B)$ and $e' \!=\! (A', B')$, we define $e < e'$ if $\max(A.x_2, B.x_2)$ $< \min(A'.x_1, B'.x_1)$ and $\max(A.y_2, B.y_2) < \min(A'.y_1, B'.y_1)$. We say that $e$ is \textit{smaller} than $e'$ and $e'$ is \textit{greater} than $e$.
\end{definition}

	In other words, $e = (A, B) < e' = (A', B')$ if and only if the big rectangle $C'$ that covers $A'$ and $B'$ is completely at the \textit{top-right} of the big rectangle $C$ that covers $A$ and $B$ in \texttt{SPACE}$(\tau)$.

\begin{definition}
	A sequence of matches $e_1, e_2, ..., e_k$ is called a \textit{chain} if $e_i < e_{i+1}$ for all $1 \leq i < k$. The \textit{length} of the chain is $k$, and the match $e_1$ is the \textit{smallest} match of the chain.
\end{definition}

	Since a chain corresponds to a maximum independent set in $L(G)^2$ and vice versa, the MIM problem in trapezoid graphs turns out to be the problem of finding the longest chain on \texttt{SPACE}$(\tau)$. We introduce an $\mathcal{O}(m + n)$ algorithm for MIM problem in $G(\tau)$ where $n = |V| = |\tau|$ and $m = |E|$ using the same idea but a little bit more tricky than the permutation one. As well as in the permutation case, we first build adjacent lists from pairs of intersecting trapezoids. Then, we calculate $f$ and $link$ functions for all these pairs. Finally, we construct a MIM based on the functions we calculated.
	
\begin{framed}
	\textbf{Steps to find a MIM in $G(\tau)$:}

	1. Construct all matches that exist on \texttt{SPACE}$(\tau)$ and store them in adjacent lists.
	
	2. Calculate two functions $f$ and $link$ for all matches $e$.
	
	3. Build a MIM based on functions $f$ and $link$ calculated in step 2.
\end{framed}

\subsubsection{Construct all matches on \texttt{SPACE}$(\tau)$} \label{SSS:trapezoid step 1}
	
	For each trapezoid $A$, let $Match(A)$ be the list of all trapezoids $B$ so that $(B, A)$ is a match. Trapezoids in $Match(A)$ are sorted in decreasing order of their $y_2$ coordinate. It is noticeable that two trapezoids $A$, $B$ intersect if and only if there is a diagonal of $A$ intersects with a diagonal of $B$. Since $x$- and $y$-coordinates of all trapezoids are distinct as we assumed, all trapezoid diagonals form a permutation graph, so $Match$ lists can be constructed by a nearly similar algorithm written in procedure \texttt{buildAllMatches} from Section \ref{SSS:permutation step construct}. We present procedure \texttt{buildAllMatches} for the trapezoid case as below. For readability and simplicity, we only show a brief version of Procedure \ref{Algo:buildAllMatches-Trapezoid} here. For the complete version of Procedure \ref{Algo:buildAllMatches-Trapezoid}, please see \textit{Appendix D}.
	
\begin{algorithm}[ht]\small
\caption{$\mathtt{buildAllMatches}(\tau)$}
\label{Algo:buildAllMatches-Trapezoid}

\SetKwInOut{Description}{Description}
\Description{\textit{Step 1:} Construct all matches on \texttt{SPACE}$(\tau)$ and store them as adjacent lists.}

\ShowLn \tcc{initialize}

Create linked list \texttt{LL} of $n$ nodes numbered from 1 to $2n$, where $\mathtt{LL}.head = 2n$, and $i.next = i - 1$ for all $1 < i \le 2n$\;
set $revMatch(A) \leftarrow \emptyset$ and $Match(A) \leftarrow \emptyset$ for all $1 \leq x \leq 2n$\;
\;
\ShowLn \tcc{build reversed match lists with possibly duplicate elements}
\For{$y \leftarrow 2n$ down to $1$}{
	$p \leftarrow \mathtt{LL}.head$\;
	suppose that $y = R.y_1$ or $y = R.y_2$ for some $R \in \tau$\;
	\While{$true$}{
		suppose that $p = A.x_1$ or $p = A.x_2$ for some $A \in \tau$\;
		\If{$A = R$}{
			\If{($p = R.x_1$ and $y = R.y_2$) or ($p = R.x_2$ and $y = R.y_1$)}{
				\textbf{break}\;
			}
		}
		\lElseIf{$A.x_2 < R.x_2$}{
			add $R$ to $revMatch(A)$
		}
		\lElse{
			add $A$ to $revMatch(R)$
		}
		$p \leftarrow p.next$\;
	}
	remove node $p$ from $\mathtt{LL}$\;
}
\;
build $Match$ lists from $revMatch$ lists\;
\Return {$\{Match(A)\ |\ A \in \tau\}$}\;
\end{algorithm}

\begin{algorithm}[!b]\small
\caption{$\mathtt{calculateFAndLink}(\tau, Match)$}
\label{Algo:calculateFAndLink-Trapezoid}

\SetKwInOut{Description}{Description}
\Description{\textit{Step 2:} Calculate 2 functions $f$ and $link$ for all matches $e$.}

\ShowLn \tcc{initialize}	
initialize $S$, $L$ and $d$\;
\ShowLn \tcc{calculate functions $f$ and $link$}
\For{$x \leftarrow 2n$ down to $1$}{
	\If{$x = A.x_2\ $ for some $A \in \tau$}{
		calculate $f$ and $link$ for all matches $(*, x)$ with pointer $z$ decreasing from $\varphi(2n+1)$\;
		update the $L$ and $d$ before $L$ passes coordinate $i = x$\;
	}
	\Else(\tcp*[f]{$x = A.x_1$ for some $A \in \tau$}){
		update elements in $S_x$ to $L$\;
	}
}
\Return {$(f, link)$} \;
\end{algorithm}

	As we assumed, $x$-coordinates and $y$-coordinates of all trapezoids are distinct and lie inside the range $[1, 2n]$. Hence we can make an $\mathcal{O}(n)$ pre-process step to find  in $\mathcal{O}(1)$ the trapezoid to which a $x$- or $y$-coordinate belongs. Although a trapezoid $B$ could appear many times in a list $revMatch(A)$, the number of $B$'s appearances is at most 4 since $A$ and $B$ have only two diagonals each. Therefore, the number of elements in all $revMatch$ lists is at most $4m$, which leads to $\mathcal{O}(m + n)$ running time of procedure \texttt{buildAllMatches}.

\subsubsection{Calculate two functions $f$ and $link$} \label{SSS:trapezoid step 2}
	
	The procedure \texttt{calculateFAndLink} (see Procedure \ref{Algo:calculateFAndLink-Trapezoid}) can be reused from Section \ref{SSS:permutation step calculate} by applying some minor adjustments without changing the time complexity (see Figure \ref{F:O(m + n) trapezoid G(tau)}). For readability and simplicity, we only show a brief version of Procedure \ref{Algo:calculateFAndLink-Trapezoid} here. For the complete version of Procedure \ref{Algo:calculateFAndLink-Trapezoid}, please see \textit{Appendix E}.
	
\begin{figure}[!htbp]
\centering
\includegraphics[width=9.6cm]{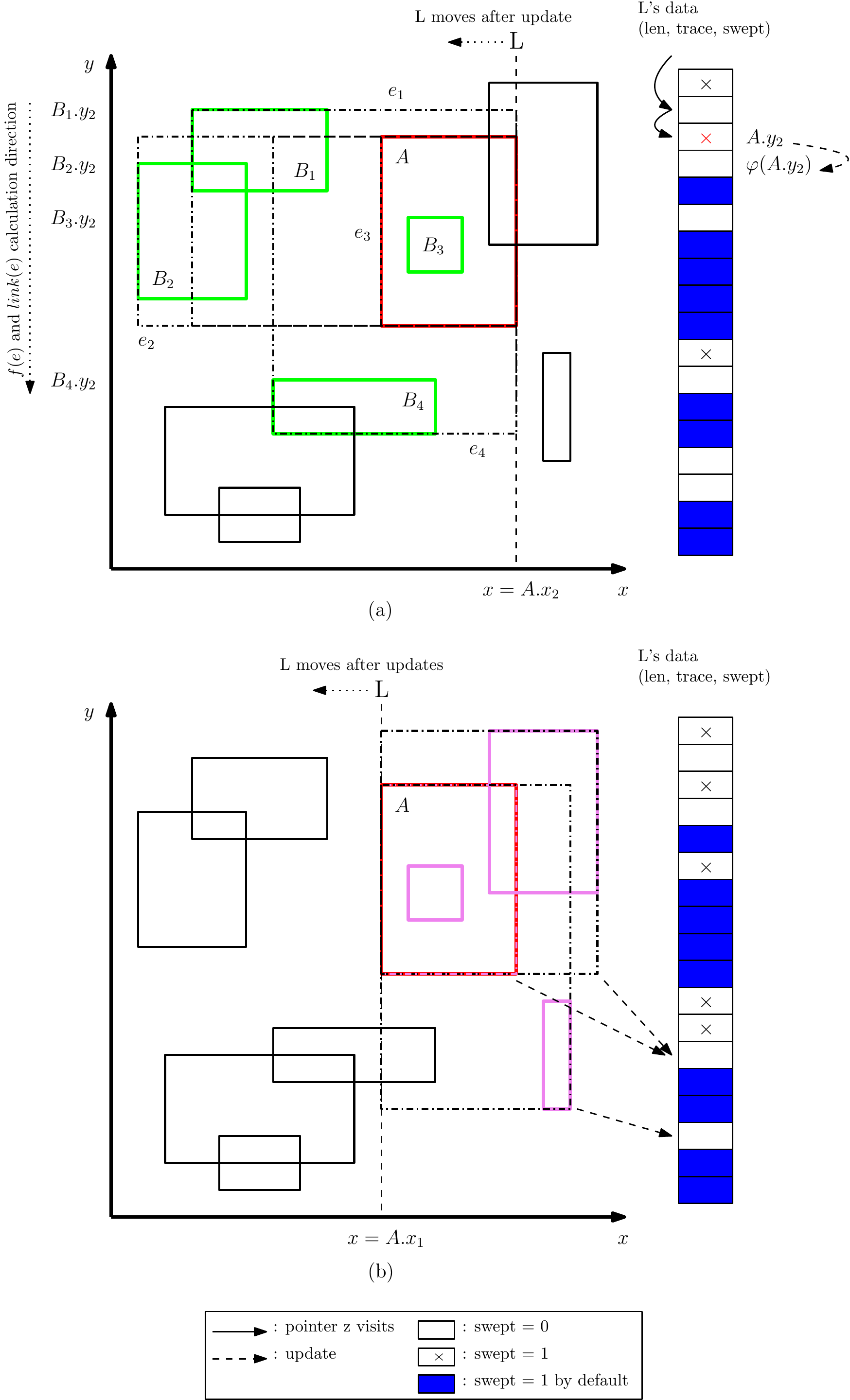}
\caption{Calculate $f$ and $link$ in $G(\tau)$ in $\mathcal{O}(m + n)$ time.\\
(a) For all matches $e = (B, A)$ where $A.x_2 = x$, $f(e)$ and $link(e)$ are calculated in decreasing order of $B.y_2$. To obtain $f(e)$ and $link(e)$, every $L_z$ such that $\max(A.y_2, B.y_2) \leq z \leq 2n$ and $L_z.swept = 0$ is visited. Then $(f(e), e)$ is added to the list $S_{\min(A.x_1, B.x_1)}$. After all, $L_{A.y_2}.swept$ is set to $1$; $L_{A.y_2}.len$ and $L_{A.y_2}.trace$ are updated to $L_{\varphi(A.y_2)}$; and the sweep line $L$ moves to the coordinate $i = x - 1$.\\
(b) Every element $(f(e), e)$ in $S_x$ having $e = (A, B)$ and $A.x_1 = x$ is updated to the cell $L_{\varphi(\min(A.y_1, B.y_1))}$. After all, the sweep line $L$ moves to the coordinate $i = x - 1$.}
\label{F:O(m + n) trapezoid G(tau)}
\end{figure}

\subsubsection{Build a maximum induced matching} \label{SSS:trapezoid step 3}

	The last procedure \texttt{buildMIM} from Section \ref{SSS:permutation step build} can be completely reused.

\subsubsection{Summary} \label{SSS:trapezoid summary}

	All three main procedures have $\mathcal{O}(m + n)$ time complexity. Therefore, given $\tau$, we can find a MIM of trapezoid graph $G(\tau)$ in $\mathcal{O}(m + n)$ time by calling the procedure \texttt{maxInducedMatching} below.\vspace*{5mm}

\begin{algorithm}[H]
\caption{$\mathtt{maxInducedMatching}(\tau)$}
\label{Algo:maxInducedMatching-Trapezoid}

\SetKwInOut{Description}{Description}
\Description{Finding a MIM in $G(\tau)$}

$Match \leftarrow \mathtt{buildAllMatches}(\tau)$\;
$(f, link) \leftarrow \mathtt{calculateFAndLink}(\tau, Match)$\;
\Return {$\mathtt{buildMIM}(f, link)$}\;

\end{algorithm}
	
	\smallskip\medskip
We summarize our approach in Theorem \ref{T:O(m + n) trapezoid G(tau)}.

\begin{theorem} \label{T:O(m + n) trapezoid G(tau)}
	A maximum induced matching in trapezoid graph $G(\tau)$ can be found in $\mathcal{O}(m + n)$~time.
\end{theorem}

\begin{proof}
    We can use the same technique in the proof of Theorem \ref{T:O(m + n) permutation G(pi)} to show that the overall running time of Procedure \ref{Algo:calculateFAndLink-Trapezoid} is $\mathcal{O}(m + n)$. In addition, since Procedure \ref{Algo:buildAllMatches-Trapezoid} runs in $\mathcal{O}(m + n)$ time and Procedure \texttt{buildMIM} runs in $\mathcal{O}(n)$ time, Procedure \ref{Algo:maxInducedMatching-Trapezoid} can produce a maximum induced matching for $G(\tau)$ in $\mathcal{O}(m + n)$ time.
\end{proof}

    By exploiting Cogis' result \cite{cogis1982ferrers} and matrix multiplication properties, Ma and Spinrad \cite{ma19942, ma1990algorithms, ma1990avoiding} created an algorithm that can recognize whether an undirected graph is a trapezoid graph or not in $\mathcal{O}(n^2)$. This algorithm is by far the fastest trapezoid graph recognition algorithm, and it can be modified to give a trapezoid model in $\mathcal{O}(n^2)$.

    Based on this result and Theorem \ref{T:O(m + n) trapezoid G(tau)}, we conclude the section by the following corollary.

\begin{corollary} \label{T:O(n^2) trapezoid G}
	A maximum induced matching in a trapezoid graph $G$ can be found in $\mathcal{O}(n^2)$ time.
\end{corollary}

\section{Conclusion}

	We have introduced efficient algorithms for the MIM problem in both permutation graphs and trapezoid graphs based on the combined technique of dynamic programming and geometrical sweep line. This method is promising to apply to optimization problems in various special graph classes. We finally summarize the main results in Table \ref{tab:Complexity to find a MIM}.

\begin{table}[H]
  \begin{center}
   \caption{Time complexity to find a MIM in particular graph classes.}
    \label{tab:Complexity to find a MIM}\vspace*{-1mm}
    \begin{tabular}{|c|c|c|}
      \toprule %
      \backslashbox{\textbf{Class of $G$}}{\textbf{Input for $G$}} & $\pi$ \textbf{or} $\tau$ & $V$ \textbf{and} $E$\\
      \midrule %
Permutation graph& \multirow{2}{*}{$\mathcal{O}(m + n)$} &$\mathcal{O}(m + n)$\\
\cmidrule{1-1} \cmidrule{3-3}
Trapezoid graph& &$\mathcal{O}(n^2)$\\
\bottomrule
    \end{tabular}
  \end{center}
\end{table}

\subsection*{Acknowledgment}
This research is funded by Vietnam National Foundation for Science and Technology Development (NAFOSTED) under grant number 102.01-2019.302.

Viet Dung Nguyen was funded by Vingroup Joint Stock Company and supported by the Domestic Master/ PhD Scholarship Programme of Vingroup Innovation Foundation (VINIF), Vingroup Big Data Institute (VINBIGDATA), code VINIF.2020.ThS.BK.05. 


\vfil\eject

\section*{Appendix A.}

\vspace*{-10.5cm}

\setcounter{algocf}{0}
\begin{algorithm}
\caption{$\mathtt{buildAllMatches}(\pi)$}

\SetKwInOut{Description}{Description}
\Description{\textit{Step 1:} Construct all matches that exist on \texttt{SPACE}$(\pi^{-1})$ and store them as adjacent lists.}

\ShowLn \tcc{initialize}
Create linked list \texttt{LL} of $n$ nodes numbered from 1 to $n$, where $\mathtt{LL}.head = n$, and $i.next = i - 1$ for all $1 < i \le n$ (the next node of node numbered 1 is \texttt{NULL})\;
\For{$x \leftarrow 1$ to $n$}{
	$Match(x) \leftarrow \emptyset$\;
}
\;
\ShowLn \tcc{build $Match$ lists}
\For{$a \leftarrow n$ down to $1$}{
	$p' \leftarrow \mathtt{NULL}$ \tcp*[r]{previous node}
	$p \leftarrow \mathtt{LL}.head$ \tcp*[r]{current node}
	\While{$p > \pi(a)$}{
		add $\pi(a)$ to $Match(p)$\;
		$p' \leftarrow p$\;
		$p \leftarrow p.next$\;
	}
	\ShowLn \tcc{remove node $p$}
	\If{$p' \neq \mathtt{NULL}$}{
		$p'.next \leftarrow p.next$\;
	}
	\Else{
		$\mathtt{LL}.head \leftarrow p.next$\;
	}
}
\;
\Return {$\{Match(x)\ |\ 1 \leq x \leq n\}$}\;

\end{algorithm}

\eject

\section*{Appendix B.}
\setcounter{algocf}{1}
\begin{algorithm}
\caption{$\mathtt{calculateFAndLink}(\pi^{-1}, Match)$}

\SetKwInOut{Description}{Description}
\Description{\textit{Step 2:} Calculate 2 functions $f$ and $link$ for all matches $e$.}

\ShowLn \tcc{initialize $S$ and $L$}
\For{$x \leftarrow 1$ to $n$}{
	$S_x \leftarrow \emptyset$\;
	$L_x \leftarrow (0, \mathtt{NULL})$\;
}
\;
\ShowLn \tcc{calculate functions $f$ and $link$}
\For{$x \leftarrow n$ down to $1$}{
	\ShowLn \tcc{calculate $f$ and $link$ for all matches $(*, x)$ with pointer $z$ decreasing from $n$}
	$z \leftarrow n$\;
	$maxLen \leftarrow 0$\;
	$trace \leftarrow \mathtt{NULL}$\;
	\For{$y \leftarrow Match(x).begin$ to $Match(x).end$}{
		\While{$z > \pi^{-1}(y)$}{
			\If{$L_z.len > maxLen$}{
				$maxLen \leftarrow L_z.len$\;
				$trace \leftarrow L_z.trace$\;
			}
			$z \leftarrow z - 1$\;
		}
		$e \leftarrow (y, x)$\;
		$f(e) \leftarrow maxLen + 1$\;
		$link(e) \leftarrow trace$\;
		add $(f(e), e)$ to the list $S_y$\;
	}
	\;
	\ShowLn \tcc{update elements in $S_x$ to $L$}
	\For{\textbf{each} $(f(e), e) \in S_x$ where $e = (x, a)$}{
		\If{$L_{\pi^{-1}(a)}.len < f(e)$}{
			$L_{\pi^{-1}(a)} \leftarrow (f(e), e)$\;
		}
	}
}
\;
\Return {$(f, link)$}\;

\end{algorithm}

\clearpage

\setcounter{algocf}{7}
\section*{Appendix C.}

\begin{algorithm}[H]
\caption{$\mathtt{calculateFAndLink}(\pi^{-1}, Match)$}

\SetKwInOut{Description}{Description}
\Description{\textit{Step 2:} Calculate 2 functions $f$ and $link$ for all matches $e$.}

\ShowLn \tcc{initiallize $S$, $L$ and $d$}
$d \leftarrow \{\{1\}, \{2\}, ..., \{n\}\}$\;
\For{$x \leftarrow 1$ to $n$}{
	$S_x \leftarrow \emptyset$,
	$L_x \leftarrow (0, \mathtt{NULL}, 0)$\;
}
\;
\ShowLn \tcc{calculate functions f and link}
\For{$x \leftarrow n$ down to $1$}{
	\ShowLn \tcc{calculate $f$ and $link$ for all matches $(*, x)$ with pointer $z$ decreasing from $\varphi(n+1)$}
	$z \leftarrow \mathtt{cal\varphi}(L, d, n + 1)$,
	$maxLen \leftarrow 0$,
	$trace \leftarrow \mathtt{NULL}$\;
	\For{$y \leftarrow Match(x).begin$ to $Match(x).end$}{
		\While{$z \geq \pi^{-1}(y)$}{
			\If{$L_z.len > maxLen$}{
				$maxLen \leftarrow L_z.len$,
				$trace \leftarrow L_z.trace$\;
			}
			$z \leftarrow \mathtt{cal\varphi}(L, d, z)$\;
		}
		$e \leftarrow (y, x)$\;
		$f(e) \leftarrow maxLen + 1$,
		$link(e) \leftarrow trace$\;
		add $(f(e), e)$ to the list $S_y$\;
	}
	\;
	\ShowLn \tcc{update elements of $S_x$ to $L$}
	\For{\textbf{each} $(f(e), e) \in S_x$ where $e = (x, a)$}{
		$b \leftarrow \mathtt{cal\varphi}(L, d, \pi^{-1}(a))$\;
		\If{$b > 0$ and $L_b.len < f(e)$}{
			$L_b \leftarrow (f(e), e, 0)$\;
		}
	}
	\;
	\ShowLn \tcc{update $L$ and $d$ before $L$ passes coordinate $i = x$}
	$L_{\pi^{-1}(x)}.swept \leftarrow 1$\;
	\If{$\pi^{-1}(x) > 1$ and $L_{\pi^{-1}(x)-1}.swept = 1$}{
		$\mathtt{union}(d, \pi^{-1}(x) - 1, \pi^{-1}(x))$\;
	}
	\If{$\pi^{-1}(x) < n$ and $L_{\pi^{-1}(x)+1}.swept = 1$}{
		$\mathtt{union}(d, \pi^{-1}(x), \pi^{-1}(x) + 1)$\;
	}
	$b \leftarrow \mathtt{cal\varphi}(L, d, \pi^{-1}(x))$\;
	\If{$b > 0$ and $L_b.len < L_{\pi^{-1}(x)}.len$}{
		$L_b \leftarrow (L_{\pi^{-1}(x)}.len, L_{\pi^{-1}(x)}.trace, 0)$\;
	}
}
\;
\Return {$(f, link)$}\;

\end{algorithm}

\setcounter{algocf}{8}
\section*{Appendix D.}

\begin{algorithm}[H]
\caption{$\mathtt{buildAllMatches}(\tau)$}

\SetKwInOut{Description}{Description}
\Description{\textit{Step 1:} Construct all matches on \texttt{SPACE}$(\tau)$ and store them as adjacent lists.}

\ShowLn \tcc{initiallize}
Create linked list \texttt{LL} of $n$ nodes numbered from 1 to $2n$, where $\mathtt{LL}.head = 2n$, and $i.next = i - 1$ for all $1 < i \le 2n$ (the next node of node numbered 1 is \texttt{NULL})\;
\For{\textbf{each} $A \in \tau$}{
	$revMatch(A) \leftarrow \emptyset$ \tcp*[r]{$A$'s reversed match list with duplicable elements}
	$Match(A) \leftarrow \emptyset$ \tcp*[r]{$Match$ list of $A$}
}
\;
\ShowLn \tcc{build reversed match lists with possibly duplicate elements}
\For{$y \leftarrow 2n$ down to $1$}{
	$p' \leftarrow \mathtt{NULL}$ \tcp*[r]{previous node}
	$p \leftarrow \mathtt{LL}.head$ \tcp*[r]{current node}
	suppose that $y = R.y_1$ or $y = R.y_2$ for some $R \in \tau$\;
	\While{$true$}{
		suppose that $p = A.x_1$ or $p = A.x_2$ for some $A \in \tau$\;
		\If{$A = R$}{
			\If{($p = R.x_1$ and $y = R.y_2$) or ($p = R.x_2$ and $y = R.y_1$)}{
				\textbf{break}\;
			}
		}
		\ElseIf{$A.x_2 < R.x_2$}{
			add $R$ to $revMatch(A)$\;
		}
		\Else{
			add $A$ to $revMatch(R)$\;
		}
		$p' \leftarrow p$,
		$p \leftarrow p.next$\;
	}
	\ShowLn \tcc{remove node $p$}
	\If{$p' \neq \mathtt{NULL}$}{
		$p'.next \leftarrow p.next$\;
	}
	\Else{
		$\mathtt{LL}.head \leftarrow p.next$\;
	}
}
\;
\ShowLn \tcc{build $Match$ lists from $revMatch$ lists}
\For{$y \leftarrow 2n$ down to $1$}{
	\If{$y = A.y_2$ for some $A \in \tau$}{
		\For{\textbf{each} $B \in revMatch(A)$}{
			\If{$A$ is not added to $Match(B)$}{
				add $A$ to $Match(B)$\;
			}
		}
	}
}
\;
\Return {$\{Match(A)\ |\ A \in \tau\}$}\;

\end{algorithm}

\setcounter{algocf}{9}
\section*{Appendix E.}
\begin{algorithm}[H]
\caption{$\mathtt{calculateFAndLink}(\tau, Match)$}

\SetKwInOut{Description}{Description}
\Description{\textit{Step 2:} Calculate 2 functions $f$ and $link$ for all matches $e$.}

\ShowLn \tcc{initiallize $S$, $L$ and $d$}
$d \leftarrow \{\{1\}, \{2\}, ..., \{2n\}\}$\;
\For{$x \leftarrow 1$ to $2n$}{
	$S_x \leftarrow \emptyset$\;
	\lIf{$x = A.y_2\ $ for some $A \in \tau$}{
		$L_x \leftarrow (0, \mathtt{NULL}, 0)$
	}
	\Else{
		$L_x \leftarrow (0, \mathtt{NULL}, 1)$\;
		\lIf{$x > 1$ and $L_{x-1}.swept = 1$}{
			$\mathtt{union}(d, x - 1, x)$
		}
	}
}
\;
\ShowLn \tcc{calculate functions $f$ and $link$}
\For{$x \leftarrow 2n$ down to $1$}{
	\If{$x = A.x_2\ $ for some $A \in \tau$}{
		\ShowLn \tcc{calculate $f$ and $link$ for matches $(*, A)$, pointer $z$ decreasing from $\varphi(2n+1)$}
		$z \leftarrow \mathtt{cal\varphi}(L, d, 2n + 1)$, $maxLen \leftarrow 0$, $trace \leftarrow \mathtt{NULL}$\;
		\For{$B \leftarrow Match(A).begin$ to $Match(A).end$}{
			\While{$z \geq \max(A.y_2, B.y_2)$}{
				\lIf{$L_z.len > maxLen$}{
					$maxLen \leftarrow L_z.len$, $trace \leftarrow L_z.trace$
				}
				$z \leftarrow \mathtt{cal\varphi}(L, d, z)$\;
			}
			$e \leftarrow (B, A)$\;
			$f(e) \leftarrow maxLen + 1$, $link(e) \leftarrow trace$\;
			add $(f(e), e)$ to the list $S_{\min(A.x1, B.x1)}$\;
		}
		\;
		\ShowLn \tcc{update $L$ and $d$ before $L$ passes coordinate $i = x$}
		$L_{A.y_2}.swept \leftarrow 1$\;
		\lIf{$A.y_2 > 1$ and $L_{A.y_2-1}.swept = 1$}{
			$\mathtt{union}(d, A.y_2 - 1, A.y_2)$
		}
		\lIf{$A.y_2 < 2n$ and $L_{A.y_2+1}.swept = 1$}{
			$\mathtt{union}(d, A.y_2, A.y_2 + 1)$
		}
		$z \leftarrow \mathtt{cal\varphi}(L, d, A.y_2)$\;
		\lIf{$z > 0$ and $L_z.len < L_{A.y_2}.len$}{
			$L_z \leftarrow (L_{A.y_2}.len, L_{A.y_2}.trace, 0)$
		}
	}
	\Else(\tcp*[f]{$x = A.x_1$ for some $A \in \tau$}){
		\ShowLn \tcc{update elements in $S_x$ to $L$}
		\For{\textbf{each} $(f(e), e) \in S_x$ where $e = (A, B)$}{
			$z \leftarrow \mathtt{cal\varphi}(L, d, \min(A.y1, B.y1))$\;
			\If{$z > 0$ and $L_z.len < f(e)$}{
				$L_z \leftarrow (f(e), e, 0)$\;
			}
		}
	}
}
\;
\Return {$(f, link)$}\;

\end{algorithm}

\end{document}